\documentclass[10pt,conference]{IEEEtran}
\usepackage{amsthm}
\usepackage{amsmath}
\usepackage{amsfonts}
\usepackage{graphicx}
\usepackage{latexsym}
\usepackage{amssymb}
\usepackage{stmaryrd}
\usepackage{comment}
\usepackage{amscd}
\usepackage{hyperref}
\newcommand\myshade{70} 
\hypersetup{ 
	linkcolor  = red!\myshade!black,
	citecolor  = blue!\myshade!black,
	urlcolor   = blue!\myshade!black,
	colorlinks = true,
}
\usepackage{cite}
\usepackage{cleveref}
\usepackage[belowskip=-15pt,aboveskip=0pt,small]{caption}
\usepackage[dvipsnames]{xcolor}
\usepackage[linesnumbered,ruled]{algorithm2e}
\usepackage[noend]{algpseudocode}
\allowdisplaybreaks
\usepackage{graphicx}
\usepackage{subfigure}
\usepackage{wrapfig}
\usepackage{float}
\usepackage{bm}
\usepackage{bbm}
\usepackage{enumerate}

\usepackage{thm-restate}

\usepackage{mathtools}
\newcommand{\bea}{\begin{eqnarray}}
\newcommand{\eea}{\end{eqnarray}}
\newcommand{\bean}{\begin{eqnarray*}}
\newcommand{\eean}{\end{eqnarray*}}
\newcommand{\ceil}[1]{\left\lceil #1 \right\rceil}
\newcommand{\floor}[1]{\left\lfloor #1 \right\rfloor}

\newcommand{\sbinom}[2]{\left( \begin{array}{c} #1 \\ #2 \end{array} \right) }

\newcommand{\field}[1]{\mathbb{#1}}

\newcommand{\F}{\field{F}}

\newcommand{\cC}{{\cal C}}

\newcommand{\cL}{{\cal L}}

\newcommand{\cO}{{\cal O}}

\newcommand{\cS}{{\cal S}}

\newcommand{\sG}{\script{G}}

\newcommand{\sP}{\script{P}}

\newcommand{\bfc}{{\boldsymbol c}}

\newcommand{\bfu}{{\boldsymbol u}}
\newcommand{\bfv}{{\boldsymbol v}}
\newcommand{\bfw}{{\boldsymbol w}}

\newcommand{\bfy}{{\boldsymbol y}}
\newcommand{\bfz}{{\boldsymbol z}}

\newcommand{\FF}{\mathbb{F}}

\DeclareMathAlphabet{\mathbfsl}{OT1}{cmr}{bx}{it}
\newcommand{\uuu}{\kern-1pt\mathbfsl{u}\kern-0.5pt}
\newcommand{\vvv}{\kern-1pt\mathbfsl{v}\kern-0.5pt}

\newcommand{\myboxplus}{\kern1pt\mbox{\small$\boxplus$}}

\makeatletter \DeclareRobustCommand{\sbinom}{\genfrac[]\z@{}}
\makeatother
\newcommand{\G}[2]{\sbinom{{#1}\kern-1pt}{{#2}\kern-1pt}}
\newcommand{\Gq}[2]{\sbinom{{#1}\kern-0.25pt}{{#2}\kern-0.25pt}}
\newcommand{\Fq}{\smash{{\mathbb F}_{\!q}}}

\newcommand{\Ps}{\smash{{\sP\kern-2.0pt}_q\kern-0.5pt(n)}}
\newcommand{\sPs}{\smash{{\sP\kern-1.5pt}_q(n)}}
\newcommand{\Ptwo}{\smash{{\sP\kern-2.0pt}_2\kern-0.5pt(n)}}
\newcommand{\Ptwom}{\smash{{\sP\kern-2.0pt}_2\kern-0.5pt(m)}}
\newcommand{\Ptwonm}{\smash{{\sP\kern-2.0pt}_2\kern-0.5pt(n+m)}}
\newcommand{\Ptwoa}{\smash{{\sP\kern-2.0pt}_2\kern-0.5pt(1)}}
\newcommand{\Ptwob}{\smash{{\sP\kern-2.0pt}_2\kern-0.5pt(2)}}
\newcommand{\Ptwoc}{\smash{{\sP\kern-2.0pt}_2\kern-0.5pt(3)}}
\newcommand{\Ptwod}{\smash{{\sP\kern-2.0pt}_2\kern-0.5pt(4)}}
\newcommand{\Ptwoe}{\smash{{\sP\kern-2.0pt}_2\kern-0.5pt(5)}}
\newcommand{\Ptwof}{\smash{{\sP\kern-2.0pt}_2\kern-0.5pt(6)}}
\newcommand{\Ptwokm}{\smash{{\sP\kern-2.0pt}_2\kern-0.5pt(2k-1)}}
\newcommand{\Pone}{\smash{{\sP\kern-2.5pt}_2\kern-0.5pt(n{-}1)}}

\newcommand{\Gr}{\smash{{\sG\kern-1.5pt}_q\kern-0.5pt(n,k)}}
\newcommand{\Gi}{\smash{{\sG\kern-1.5pt}_q\kern-0.5pt(n,i)}}
\newcommand{\Gj}{\smash{{\sG\kern-1.5pt}_q\kern-0.5pt(n,j)}}
\newcommand{\Grmk}{\smash{{\sG\kern-1.5pt}_q\kern-0.5pt(n,n-k)}}
\newcommand{\Grdk}{\smash{{\sG\kern-1.5pt}_q\kern-0.5pt(2k,k)}}
\newcommand{\Grekappa}{\smash{{\sG\kern-1.5pt}_q\kern-0.5pt(n,e+1-\kappa)}}
\newcommand{\Grtwoekappa}{\smash{{\sG\kern-1.5pt}_q\kern-0.5pt(n,2e+1-\kappa)}}
\newcommand{\Gremkappa}{\smash{{\sG\kern-1.5pt}_q\kern-0.5pt(n,e-\kappa)}}
\newcommand{\Gn}{\smash{{\sG\kern-1.5pt}_2\kern-0.5pt(n,n{-}1)}}
\newcommand{\Gnq}{\smash{{\sG\kern-1.5pt}_q\kern-0.5pt(n,n{-}1)}}
\newcommand{\Gone}{\smash{{\sG\kern-1.5pt}_2\kern-0.5pt(n,1)}}
\newcommand{\Gqone}{\smash{{\sG\kern-1.5pt}_q\kern-0.5pt(n,1)}}
\newcommand{\GTwo}{\smash{{\sG\kern-1.5pt}_2\kern-0.5pt(n,k)}}
\newcommand{\GTwonk}[2]{{\smash{{\sG\kern-1.5pt}_2\kern-0.5pt({#1},{#2})}}}
\newcommand{\Gnk}{\smash{{\sG\kern-1.5pt}_2\kern-0.5pt(n,n{-}k)}}
\newcommand{\Greone}{\smash{{\sG\kern-1.5pt}_q\kern-0.5pt(n,e{+}1)}}
\newcommand{\Gretwo}{\smash{{\sG\kern-1.5pt}_q\kern-0.5pt(n,e{+}2)}}

\newcommand{\be}[1]{\begin{equation}\label{#1}}
\newcommand{\ee}{\end{equation}}

\newtheorem{theorem}{Theorem}
\newtheorem{lemma}{Lemma}
\newtheorem{remark}{Remark}

\newtheorem{claim}{Claim}
\newtheorem{definition}{Definition}
\newtheorem{proposition}{Proposition}

\begin{document}


\author{\IEEEauthorblockN{ \textbf{Shubhransh~Singhvi}\IEEEauthorrefmark{1}, \textbf{Roni~Con}\IEEEauthorrefmark{3}, \textbf{Han~Mao~Kiah}\IEEEauthorrefmark{4} and \textbf{Eitan~Yaakobi}\IEEEauthorrefmark{3}\IEEEauthorrefmark{4}}
  \IEEEauthorblockA{\IEEEauthorrefmark{1}%
  Signal Processing  \&  Communications Research  Center, IIIT Hyderabad, India}
\IEEEauthorblockA{\IEEEauthorrefmark{3}%
                     Department of Computer Science, 
                     Technion---Israel Institute of Technology, 
                     Haifa 3200003, Israel}
  \IEEEauthorblockA{\IEEEauthorrefmark{4}%
                     School of Physical and Mathematical Sciences, 
		Nanyang Technological University, Singapore}
 }

\title{\textbf{An Optimal Sequence Reconstruction Algorithm for Reed-Solomon Codes}}
\date{\today}
 \maketitle
\thispagestyle{empty}	
\hspace*{-3mm}\begin{abstract}
The {\em sequence reconstruction problem}, introduced by Levenshtein in 2001, considers a scenario where the sender transmits a codeword from some codebook, and the receiver obtains $N$ noisy outputs of the codeword. 
We study the problem of {\em efficient} 
reconstruction using $N$ outputs that are corrupted by substitutions. 
Specifically, for the ubiquitous Reed-Solomon codes, we adapt the Koetter-Vardy soft-decoding algorithm, presenting a reconstruction algorithm capable of correcting beyond Johnson radius. 
Furthermore, the algorithm uses $\cO(nN)$ field operations, where $n$ is the codeword length. 
\end{abstract}

\section{Introduction}

The \emph{sequence reconstruction} problem introduced by Levenshtein \cite{L01A, L01B} corresponds to a model in which a sequence from some codebook is transmitted over several noisy channels. 
The channels are assumed to be independent, except it is required that their outputs are different. The main problem under this paradigm has been to determine the minimum number of channels required to uniquely reconstruct the transmitted sequence. Levenshtein proved that for unique reconstruction, the number of channels in the worst case has to be greater than the maximum intersection size between two balls of any possible two inputs.
Here, the ball of an input refer to all possible channel outputs of the specific input.

Also, of interest is the task to design an {\em efficient} decoder that correctly reconstructs a codeword from these noisy outputs. 
While Levenshtein introduced this problem in his seminal work \cite{L01A},
efficient decoders are less studied and to the best of authors' knowledge, only \cite{YB18, AY21, PGK22, JLL23} designed efficient decoders.

This problem was first motivated by the fields of biology and chemistry, however it is also relevant for applications in wireless sensor networks. Recently this model has received significant attention due to its applicability to DNA storage, where the same information is read multiple times and thereby several channel estimations of the data are provided~\cite{CGK12,Getal13,YKGMZM15}. Solving the reconstruction problem was studied in~\cite{L01A} with respect to several channels such as the Hamming distance, the Johnson graphs, and other metric distances. In~\cite{K08,K07,KLS07}, it was analyzed for permutations, and in~\cite{LKKM08,LS09} for other general error graphs. Later, the problem was studied in~\cite{YSLB13} for permutations with the Kendall's $\tau$ distance and the Grassmann graph, and in~\cite{GY16, SGSD17, PGK22 } for insertions and deletions, respectively. The connection between the reconstruction problem and associative memories was proposed in~\cite{YB18} and more results were then derived in~\cite{JL14,JL15,JL16}. This problem was also studied in~\cite{JV04} for the purpose of asymptotically improving the Gilbert--Varshamov bound.

In this work, we consider the $t$-substitution model, where every channel introduces at most $t$ substitution errors. To describe the problem formally, we introduce some notation. For $\bfv\in \mathbb{F}_q^n$, let $B_t(\bfv)$ be the radius-$t$ ball surrounding a word $\bfv$.  
Assume that the transmitted words belong to some code $\cC \subseteq \mathbb{F}_q^n$ with minimum distance $d$. 
Denote by $N_{n,q}(t,d) \triangleq \max_{\bfc_1,\bfc_2\in\cC} \left\vert B_t(\bfc_1) \cap B_t(\bfc_2) \right\vert$, the maximum intersection between two balls of radius $t$ of any possible pair of codewords.
In general, computing $N_{n,q}(t,d)$, is not straightforward and in some settings the exact value is not known. However, for substitution errors, Levenstein computed this quantity


\begin{lemma}[\hspace*{-1.5mm}\cite{L01A}]\label{seq_recon_bound}
Let $e\triangleq \floor{\frac{d-1}{2}}$ and $t \triangleq e + \ell$. Then 
\begin{align*}
&N_{n,q}(t, d)\\ &=\sum_{i=0}^{t-\left\lceil\frac{d}{2}\right\rceil}\binom{n-d}{i}(q-1)^i  \cdot \\
&\quad \left( \sum_{a=d-t+i}^{t-i}~\sum_{b=d-t+i}^{t-i}\binom{d}{a}\binom{d-a}{b}(q-2)^{d-a-b}\right)\;.
\end{align*}

\end{lemma}

We note that in order to read the entire input, one has to read $\cO(n\cdot N_{n,q}(t,d))$ elements in $\mathbb{F}_q$. Hence, a decoder is said to be \emph{optimal} if it takes $\cO(nN_{n,q}(t,d))$ $\mathbb{F}_q$ field operations to output the correct codeword. Note that we measure time complexity in terms of field operations over $\FF_q$. Thus, we omit $\text{poly}(\log q)$ factors in our complexity notations.



\subsection{Existing Reconstruction Decoders}
\label{sec:existing}
Let $\text{Vol}_q(\ell, n)\subset \Fq ^n$ denote the volume of the $q$-ary Hamming ball of radius $\ell$. For brevity, we let $N\triangleq N_{n,q}(t,d)+1$. Let the set of $N$ channel outputs (reads) be denoted by $Y \triangleq \left\{\bfy_{1}, \ldots, \bfy_{N}\right\} \subseteq B_{t}^{S}(\bfc)$ for some $\bfc \in \mathcal{C}$. We briefly describe possible solutions to the reconstruction problem, with their complexity.

\subsubsection{Decoder based on majority-logic-with-threshold}
 Levenshtein in \cite{L01A}, showed that the majority-logic decoder is optimal when the code $\cC$ is the entire space. However, when $\cC$ is not the entire space, a unique reconstruction decoder for $q=2$ was designed in \cite{AY19} and was recently extended to arbitrary $q$ in \cite{JLL23}. Furthermore, it was shown that the decoder has a run-time complexity of $\cO(q^{\min\{n,t(e+2)\}}nN)$, which is optimal when $q^{t(e+2)}$ is a constant. However, for $q=\cO(n)$, which is the case for RS codes, the decoder is far from optimality.

\subsubsection{Brute-force decoder}
This decoder iterates through all the codewords in $\cC$ to find the word $\bfc$ such that $Y \subseteq B_{t}(\bfc)$. Note that it takes $\cO(|\cC|Nn)$ time to find the correct word.

\subsubsection{List-decoding using a single read}
The decoder selects any one out of the $N$ reads and generates a list of codewords, $\cL$, in radius $t$ using a list decoder for the code. Next, the decoder iterates through all the codewords in the list to find the word $\bfc$ such that $Y \subseteq B_{t}(\bfc)$. Note that the run-time complexity of this decoder is the complexity that it takes to produce the list plus the time that it iterates over all the codewords in the list, which is again far from optimality. 
\vspace{1mm}

In summary, in the regime where $t=\cO(n)$ and $|C|=q^{Rn}$ for some $R>0$, both the majority-logic-with-threshold and brute-force decoders have complexities significantly higher than $\cO(n\cdot N)$.
On the other hand, the complexity of the list-decoder (using a single read) relies on the availability of an efficient list-decoding algorithm for the code $\cC$. Hence, our focus is on codes equipped with efficient list-decoding algorithms, aiming to augment the error-correcting radius in the context of the sequence reconstruction problem. 
A natural candidate is the marvelous {\em Reed-Solomon codes}. 

\subsection{Our Decoder}
 Reed-Solomon codes (RS codes)\cite{RS60} are the most widely used family of codes in theory and practice and have found many applications (some applications include QR codes, secret sharing schemes, space transmission, data storage and more. The ubiquity of these codes can be attributed to their  simplicity as well as to their efficient encoding and decoding algorithms. We next give a definition of RS codes 
 \begin{definition}\label{RS}
		Let $\alpha_1, \alpha_2, \ldots, \alpha_n \in\F_q$ be distinct points in a finite field $\mathbb{F}_q$ of order  $q\geq n$. For $k\leq n$, the $[n,k]_q$ RS code,  
		defined by the evaluation vector $\bm{\alpha} = ( \alpha_1, \ldots, \alpha_n )$, is the set of codewords 
		\[\left \lbrace c_f = \left( f(\alpha_1), \ldots, f(\alpha_n) \right) \mid f\in \mathbb{F}_q[x],\deg f < k \right \rbrace \;.\]
\end{definition}
RS codes are efficiently unique decodable up to half their minimum distance (using, e.g., the Berlekamp–Welch algorithm). 
They are also efficiently list decodable up to the Johnson radius \cite{GS99}. Namely, let $\cC$ be an $[n,k]$ RS code of rate $R:=k/n$. Then, for $\rho \leq 1 - \sqrt{R}$ (the Johnson radius), there is a polynomial time algorithm that given $y\in \F_q^n$ outputs a list of codewords $\cL$, where $\cL$ is such that (i) $d(c,y)\leq \rho n$ for all $c\in \cL$ and (ii) $|\cL| = \text{poly} (n)$.

In this paper, we provide a reconstruction decoder for RS codes that can decode from any $N$ corrupted reads at distance at most $t$ from the transmitted codeword. The running time of our decoder is $\cO(n\cdot N)$, which is the order of the input size and is thus optimal. We also show that in some settings, the value of $t$ for which our decoder works, exceeds the Johnson bound significantly. Formally, we prove the following theorem.
\begin{theorem} \label{thm:result}
    Let $\varepsilon > 0$.
    Let $\cC$ be an $[n,k]_q$ RS code and let $t$ be an integer such that
    \begin{equation} \label{eq:dec-radius-linear}
        \frac{t}{n} \leq 1 - \sqrt{\frac{k}{n} \cdot \left( 1 - \frac{\ell}{2n} + \varepsilon \right) }\;,
    \end{equation}
    where $\ell = t - \floor{\frac{d-1}{2}}$. Let $\bfc \in \cC$ be a codeword and $Y \subseteq B_t(\bfc)$ with $|Y|= N \ge N_{n,q}(t,n-k+1) + 1$.
    Then there exists an algorithm that takes the received set $Y$ as its input and outputs $\bfc$ in time $\cO(n\cdot N + n^3\varepsilon^{-6})$.
\end{theorem}

We note that by slightly changing the algorithm constructed in \Cref{thm:result}, we can improve the number of errors that we can recover from, $t$, at the expense of higher complexity. Formally,
\begin{theorem}
    Let $\varepsilon > 0$.
    Let $\cC$ be an $[n,k]_q$ RS code and let $t$ be an integer such that 
    \begin{equation} \label{eq:dec-radius-quadratic}
        \frac{t}{n} \leq 1 - \sqrt{\frac{k}{n} \cdot \left( 1 - \frac{\ell}{n} + \varepsilon \right) } \;,
    \end{equation}
    where $\ell = t - \floor{\frac{d-1}{2}}$. Let $\bfc \in \cC$ be a codeword and let $Y \subseteq B_t(\bfc)$ with $|Y|=N \ge N_{n,q}(t,n-k+1) + 1$.
    Then there exists an algorithm that takes the received set $Y$ as its input and outputs $\bfc$ in time $\cO(n\cdot N^2 + n^3 \varepsilon^{-6})$. 
\end{theorem}
\begin{remark}
    We note that our decoding radius is beyond the Johnson bound. Indeed, denote $\rho = t/n$, $R = k/n$ and observe that \eqref{eq:dec-radius-linear} becomes
    \begin{equation} \label{eq:dec-radius-linear-asym}
    \rho \leq 1 - \sqrt{R \left(1 - \left( \frac{\rho}{2} - \frac{1 - R}{4}\right) + \varepsilon \right)} \;,    
    \end{equation}
    and \eqref{eq:dec-radius-quadratic} becomes
    \begin{equation} \label{eq:dec-radius-quadratic-asym}
    \rho \leq 1 - \sqrt{R \left(1 - \left( \rho - \frac{1 - R}{2}\right) + \varepsilon \right)} \;.        
    \end{equation}
    A graphical comparison with the Johnson radius is given in \Cref{fig:john-comp}.
\end{remark}
\begin{remark}
    We compare the performance of our decoder with the one suggested above that performs list decoding on a single read.
    List decoding up to \emph{almost} the Johnson radius, i.e., up to $1 - \sqrt{R (1 + \varepsilon)}$ takes $\text{poly}(n, 1/\varepsilon)$ time and produces a list of constant size \cite{GS99}. Thus, the total complexity of this algorithm is $\cO(n\cdot N + \text{poly}(n, 1/\varepsilon))$, as in \Cref{thm:result}. However, this solution works only for values of $t$ such that $t/n \leq 1 - \sqrt{ (k/n) \cdot  (1 + \varepsilon)}$ whereas our decoder works for larger values of $t$, as implied by \eqref{eq:dec-radius-linear}.  
\end{remark}
\begin{figure}
    \centering
    \includegraphics[scale=0.45]{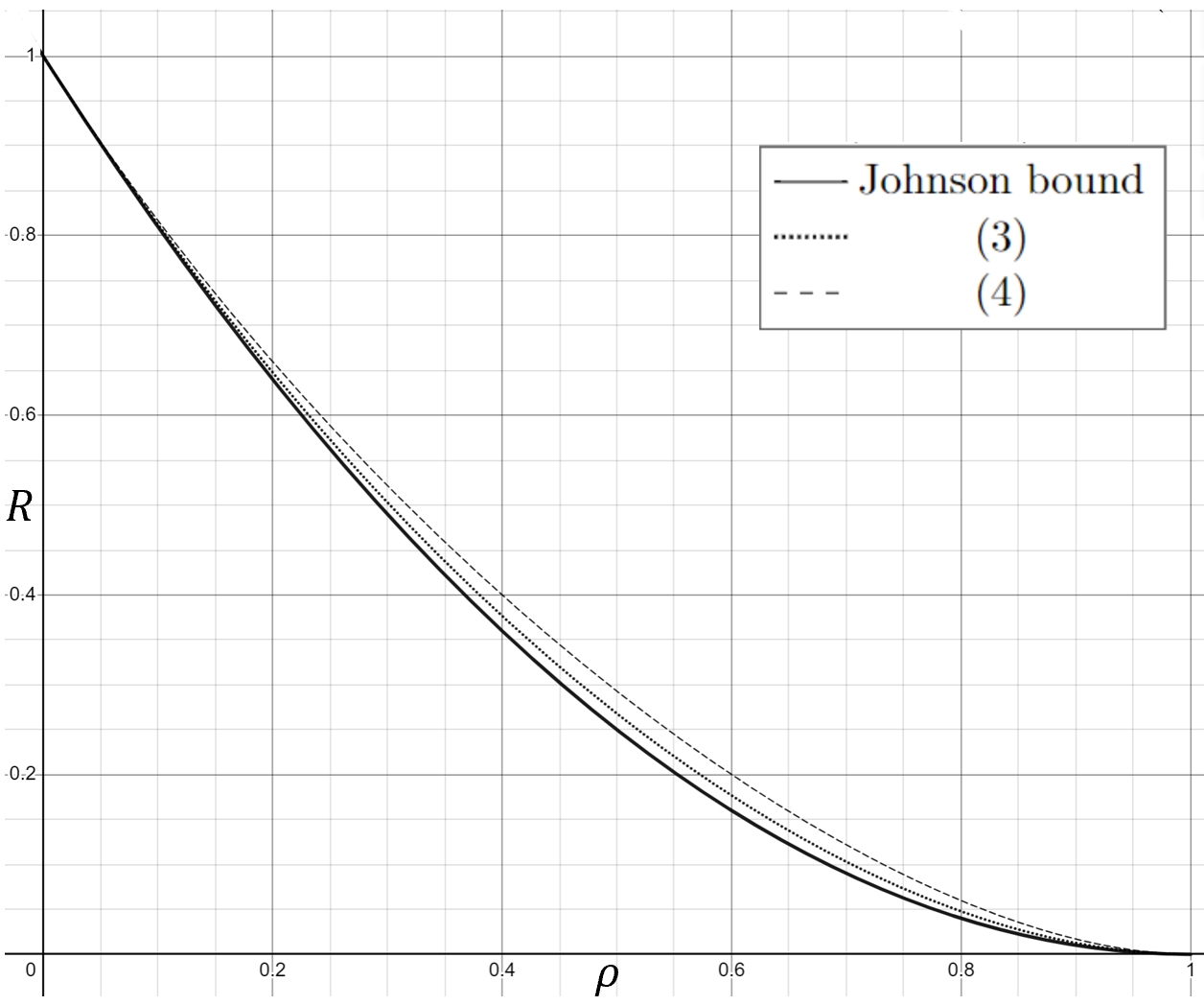}
    \caption{Tradeoff between rate $R$ and the fraction of errors that can be corrected. The algorithm that achieves the tradeoff corresponding to~\eqref{eq:dec-radius-linear-asym} has complexity $\cO(nN)$, while the algorithm achieving~\eqref{eq:dec-radius-quadratic-asym} has complexity $\cO(nN^2)$. }
    \label{fig:john-comp}
\end{figure}



\subsection{Soft Decoding \'a la Koetter and Vardy}
We briefly recall the soft-decision list-decoding algorithm of Koetter and Vardy (KV) \cite{KV03} which is an extension of the Guruswami-Sudan (GS) list decoding for RS codes. The interested reader is refered to \cite{GS99} and \cite{KV03}.

Koetter and Vardy \cite{KV03} extended the GS algorithm to the case where the decoder is supplied with probabilistic reliability information concerning the received symbols. In particular, the Koetter-Vardy algorithm 
performs a soft-decision decoding by assigning unequal multiplicities to points according to this extra information. A convenient way to keep track of the interpolation points and their different multiplicities is by means of a multiplicity matrix. 

\begin{definition}
Let $\delta_0,\delta_1, \ldots, \delta_{q-1}$ be some ordering of $\FF_q$.
A \emph{multiplicity matrix}, denoted by $M$, is a $(q \times n)$-matrix with entries $m_{i,j}$ denoting the multiplicity of $(\delta_i, \alpha_j)$.
\end{definition}

We provide a high-level description.
Given a multiplicity matrix $M$, the KV algorithm computes a non-trivial bivariate polynomial $Q_M(X,Y)$ of minimal $(1,k-1)$-weighted degree that has a zero of multiplicity at least $m_{i,j}$ at the point $(\alpha_j,\delta_i)$ for every $(i,j)$ such that $m_{i,j}\neq 0$. Then, the algorithm factorizes this polynomial to get a list of candidate codewords. (we refer again to \cite{KV03}).

The cost of constructing $Q_M(X,Y)$ for a given multiplicity matrix $M$, denoted by $C(M)$, is the number of linear equations that needs to be satisfied for the interpolation. Specifically, 

\begin{definition}
\label{def:cost}
The \emph{cost} for a multiplicity matrix $M$ is defined as follows:
\vspace{-0.3cm}
\begin{align*}
    C(M) = \sum_{i=0}^{q-1}\sum_{j=0}^{n-1}\binom{m_{i,j}+1}{2} \;.
\end{align*}

\end{definition}

For $\bfv\in\FF_q^{n}$, let $[\bfv]$ denote the $(q \times n)$-matrix representation of $\bfv$, i.e., $[\bfv]_{i,j} = 1$ if $\bfv_j = \delta_i$, and $[\bfv]_{i,j} = 0$ otherwise. 

\begin{definition}
    The score of a vector $\bfv\in\FF_q^{n}$ with respect to a given multiplicity matrix $M$ is defined as the inner product $\cS_M(\bfv) = \left<M,[\bfv]\right>.$
\end{definition}

With these definitions, we are ready to give a black box description of the KV algorithm
\begin{theorem}[\cite{KV03}] \label{thm:kv-alg}
    Let $\cC$ be an $[n,k]_q$ RS code defined with $\bm{\alpha} = ( \alpha_1, \ldots, \alpha_n )$. Let $M$ be an $q\times n$ matrix, which denotes a multiplicity matrix. 
    Given the input $M$, the KV algorithm outputs a list $\cL$ such that the following holds.
    \begin{enumerate}[(i)]
    \item A codeword $c\in \cC$ is in the list $\cL$ if we have $\cS_M(\bfc) \ge \sqrt{2(k-1)C(M)}$. 
    \item It holds that $|\cL| \leq \sqrt{\frac{2C(M)}{k-1}}$.
    \item The algorithm runs in time $\cO((C(M))^3)$.
    \end{enumerate}
    
\end{theorem}






\section{Constructing the Multiplicity Matrix in Our Settings}

In Algorithm~\ref{alg:mult-mat}, we describe how we construct the multiplicity matrix from a set of reads $Y'\subseteq Y$.  
Then, we shall compute the cost of the constructed multiplicity matrix and the score 
of a transmitted codeword with respect to the generated multiplicity matrix.
\vspace{-0.2cm}
\begin{algorithm}
    \caption{Multiplicity matrix constructor}
    \label{alg:mult-mat}
    \SetAlgoLined
    \DontPrintSemicolon
    
    \SetKwInOut{Input}{input}
    \SetKwInOut{Output}{output}

    \Input{A set $Y'$ of reads and an integer $\mu$}
    \Output{A multiplicity matrix $M\in \FF_q^{q\times n}$}
    Set $M = \textbf{0}_{q\times n}$ \;
    \For{$j \in [n]$}{
        \For{$\bfy \in Y'$}{
            Set $i$ such that $\delta_i = y_j$ \;
            $(M)_{i,j} = (M)_{i,j} + \mu$
        }
    }
    Return $M$ 
    
\end{algorithm}

\vspace{-0.5cm}
\begin{lemma}
\label{lem:score_transmitted_codeword}
    Let $M$ be the multiplicity matrix generated by Algorithm~\ref{alg:mult-mat} when given $Y'\subset Y$ and $\mu$ as input.
    Then, the score of the transmitted codeword $\bfc$ with respect to $M$ is
    $$
    \cS_M(\bfc) \geq \mu \cdot |Y'|\cdot (n-t).
    $$
\end{lemma}
\begin{proof}
    Since each word can have at most $t$ erroneous positions, the result follows.
\end{proof}

Next, we give a combinatorial lemma, whose proof we defer to the Appendix.

\begin{restatable}{lemma}{combLemma}\label{comb}
    Let $a_0, a_1, \ldots, a_{b-1}$ be positive integers such that $\sum_{i=0}^{b-1}a_i = c$. Then 
    \begin{align*}
        \sum_{i=0}^{b-1}\binom{a_i+1}{2} = \binom{c+1}{2} - \frac{1}{2}\sum_{i=0}^{b-1}a_i(c-a_i). 
    \end{align*}
\end{restatable}

Lemma~\ref{comb} then allows us to analyze the cost of the multiplicity matrix.  

\begin{lemma} \label{lem:mult-mat-cost}
    Let $M$ be the matrix returned by Algorithm~\ref{alg:mult-mat} with input $Y'$ and $\mu$. The cost of $M$ is 
    \begin{equation} \label{eq:mult-mat-cost-general}
    C(M) =  n\binom{\mu|Y'|+1}{2} - \frac{\mu^2}{2}\sum_{\bfv,\bfu \in Y'} d(\bfv,\bfu)\;.
    \end{equation}
\end{lemma}

\begin{proof}
From Definition~\ref{def:cost}, it follows that, 
\begin{align*}
    C(M) = \sum_{j=0}^{n-1} \left( \sum_{i=0}^q \binom{m_{i,j}+1}{2} \right).
\end{align*}
Note that the sum of each column of $M$ is exactly $\mu|Y'|$. Therefore, for any $j$, by \Cref{comb}, we have that 
\[    \sum_{i=0}^q \binom{m_{i,j}+1}{2} = \binom{\mu|Y'|+1}{2} - \frac{1}{2}\sum_{i=0}^{q-1}m_{i,j}(\mu|Y'|-m_{i,j}).
\]
Now, by the definition of $m_{i,j}$ and the process by which we construct the matrix, we have (recall that $\{\delta_0, \ldots, \delta_{q-1}\} = \F_q$)
{\small 
\begin{align*}
    \sum_{i=0}^{q-1}m_{i,j}(\mu|Y'|-m_{i,j}) &= \sum_{i=0}^{q-1} \left( \sum_{\bfu \in Y'} \mu \mathbbm{1}_{\bfu_j = \delta_i} \right) \left( \sum_{\bfv \in Y'} \mu \mathbbm{1}_{\bfv_j \neq \delta_i} \right) \\
    &= \mu^2 \sum_{\bfu, \bfv \in Y'} \left( \sum_{i=0}^{q-1} \mathbbm{1}_{\bfu_j = \delta_i}\cdot \mathbbm{1}_{\bfu_j \neq \delta_i} \right) \\
    &= \mu^2 \left( \sum_{\bfu, \bfv \in Y'} \mathbbm{1}_{\bfv_j \neq \bfu_j} \right),
\end{align*}
}
where $\mathbbm{1}$ is the indicator function. 
Therefore, putting it together, we have
\begin{align*}
    C(M) &= \sum_{j=0}^{n-1}\left( \binom{\mu|Y'|+1}{2} - \frac{\mu^2}{2} \sum_{\bfu, \bfv \in Y'} \mathbbm{1}_{\bfv_j \neq \bfu_j} \right)  \\
    &= n\cdot \binom{\mu|Y'|+1}{2} - \frac{\mu^2}{2} \sum_{\bfu, \bfv \in Y'} d(\bfu, \bfv) \;,
\end{align*}

\vspace{-0.1cm}
as desired.
\end{proof}

\begin{remark}
    From \eqref{eq:mult-mat-cost-general}, we observe that for a given number of reads and an integer $\mu$, the cost of the multiplicity matrix is dominated by the total distance of the reads in $Y'$. 
    Concretely, the larger the total distance is, the smaller the cost is.
\end{remark}

The reconstruction decoder is given in Algorithm~\ref{alg:rec-dec-alg}. 
\vspace{-0.2cm}
\begin{algorithm}
    \caption{Reconstruction decoder}
    \label{alg:rec-dec-alg}
    \SetAlgoLined
    \DontPrintSemicolon
    \SetKwInOut{Input}{input}
    \SetKwInOut{Output}{output}
    \Input{A set of $N$ distinct reads $Y$, a subset $Y'\subset Y$, and an integer $\mu$}
    \Output{A codeword $c$}
    Generate a multiplicity matrix $M$ using Algorithm~\ref{alg:mult-mat} with input $Y'$ and $\mu$\;
    Run the KV algorithm with the multiplicity matrix $M$\;
    Return $\bfc\in \cL$ such that $Y \subseteq B_{t}(\bfc)$ \;
\end{algorithm}
\vspace{-0.2cm}

We summarize the correctness of this algorithm in the following proposition.
\begin{proposition}
    Let $\cC$ be an $[n,k]_q$ RS code. 
    Let $\bfc \in \cC$ and 
    let $Y \subseteq B_t(\bfc)$ with $|Y|\ge N:=N_{n,q}(t,n-k+1) + 1$.
    Suppose that $M$ is the matrix generated in the first step of  Algorithm~\ref{alg:rec-dec-alg} when given inputs $Y$, $Y'$, and $\mu$. 
    Then Algorithm~\ref{alg:rec-dec-alg} outputs the codeword $\bfc$ if $\mu \cdot |Y'| \cdot (n-t) >  \sqrt{2 (k-1)C(M)}$. 
\end{proposition}
\begin{proof}
    At the first step, we construct $M$ using the reads in $Y'$ by invoking Algorithm~\ref{alg:mult-mat}. By \Cref{lem:score_transmitted_codeword}, the score of $\bfc$ is at least $\mu \cdot |Y'| \cdot (n-t)$.
    Since by our assumption, {${\mu \cdot |Y'| \cdot (n-t) >  \sqrt{2 (k-1)C(M)}}$}, according to \Cref{thm:kv-alg}, $\bfc$ is inside the list returned by the KV algorithm. Now, 
    by the Levenshtein's reconstruction problem settings, all the reads are at distance $t$ from $\bfc$ and thus, the last step uniquely identifies $\bfc$, the transmitted codeword. 
\end{proof}

\section{Constructing the Multiplicity Matrix Using only Two Reads}

In this section, we examine the performance of a multiplicity matrix constructed from only two reads. We will observe that the distance between these two reads significantly impacts the error correction capability of the KV algorithm. Specifically, a larger distance results in a greater decoding radius. In~\Cref{sec:general-two-reads}, we shall consider a general case first, where the two reads did not necessarily come from the Levenshtein's reconstruction problem. In~\Cref{sec:two-reads-lev}, we present our reconstruction decoder for the Levenshtein's reconstruction problem. 
\vspace{-0.1cm}
\subsection{Connecting the Distance between the Two Reads with the Decoding Radius}
\label{sec:general-two-reads}

In the following proposition we are applying the KV algorithm with a multiplicuty matrix based on two reads. 
We show that if we are guaranteed that the two reads are far from each other, then the correction capability exceeds the Johnson radius and that the error correction capability increases in correlation with the distance between the two reads.
Formally,

\begin{proposition} \label{prop:dec-two-reads}
    Let $\delta \in [0,2]$, $\rho\in (0,1)$, and $\varepsilon > 0$. Let $C$ be an $[n,k]_q$ RS code of rate $R = k/n$ such that 
    \begin{equation} \label{eq:rate-dist}
        1 - \rho \geq \sqrt{R \cdot \left( 1 - \frac{\delta\rho }{2} + \varepsilon \right)} \;.
    \end{equation} 
    Let $\bfv$ and $\bfu$ be two corrupted reads of a codeword $\bfc$ such that 
    \vspace{-0.1cm}
    \begin{itemize}
        \item $d(\bfv, \bfc) \leq \rho n$ and $d(\bfu, \bfc) \leq \rho n$.
        \item $d(\bfv, \bfu) = \delta \cdot \rho n$.
    \end{itemize}
    Let $M$ be the multiplicity matrix obtained by Algorithm~\ref{alg:mult-mat} with 
    input $Y = Y' = \{\bfv, \bfu\}$ and $\mu = 1/ \varepsilon$.
    Then, the KV Algorithm with multiplicity matrix $M$ outputs a list of codewords $\cL$ such that $\bfc\in \cL$ and $|\cL| = O(1/\varepsilon\sqrt{R})$.
\end{proposition}
\begin{proof}
    The condition for a codeword $\bfc$ to be in the list, according to the KV algorithm analysis is that 
    $S_{M}(\bfc) > \sqrt{2(k-1)C(M)}$. 
    According to \Cref{lem:score_transmitted_codeword}, $S_{M}(\bfc) \geq (n-\rho n) \cdot 2\mu$ and according to \Cref{lem:mult-mat-cost}, 
    \[
        C(M) = n\binom{2\mu+1}{2} - \mu^2\cdot \delta \rho n = n\mu^2 \left( 2 + \frac{1}{\mu} - \delta \rho \right) \;.
    \]
    Thus, we have to ensure that 
    \[
    (n-\rho n) \cdot 2\mu > \sqrt{2(k - 1)\cdot n \mu^2\cdot \left(2 + \frac{1}{\mu} - \delta \rho \right)}\;.
    \]
    Divide by $n\cdot 2\mu$ to get
      \begin{align*}
        1 - \rho &> \sqrt{ \left(R - \frac{1}{n} \right) \cdot \left(1 - \frac{\delta \rho}{2} + \frac{\varepsilon}{2} \right)}\;,
    \end{align*} 
    and note that this inequality clearly holds according to \eqref{eq:rate-dist}.
    As for the list size, according to \Cref{thm:kv-alg}, the list size is upper bound by 
    \[
     \sqrt{\frac{n\mu^2 \left( 4 - 2\delta\rho + \frac{2}{\mu}\right)}{k-1}} = O\left( \frac{1}{\varepsilon\sqrt{R}} \right)\;. \qedhere
    \]
\end{proof}

\begin{remark}
    We consider the two extreme points of $\delta$.
    \begin{itemize}
        \item When $\delta = 0$, the two reads are equal. In this case, inequality \eqref{eq:rate-dist} yields exactly the Johnson bound.
        \item When $\delta = 2$, inequality \eqref{eq:rate-dist} gives
        \[
        1 - \rho \geq \sqrt{R \cdot \left( 1 - \rho + \varepsilon \right)} \;,
        \]
        which implies 
        \[
        \rho < 1 - R - \varepsilon \;.
        \]
    \end{itemize}
\end{remark}

    

\subsection{Back to Levenshtein's Reconstruction Problem}
\label{sec:two-reads-lev}
In this section, we prove \Cref{thm:result}. 
In Levenstein's reconstruction problem, given a codeword $\bfc\in \cC$, an adversary will output a set $Y$ of $N$ reads such that (i) all the reads are at distance $t$ from $\bfc$ and (ii) the $N$ reads are sufficient to identify $\bfc$ uniquely. 
According to \Cref{prop:dec-two-reads}, our goal is to find two reads of maximal distance inside the set $Y$. We first start with the following lemma whose proof for the binary case ($q=2$) was given in \cite[Lemma 17]{YB18}

\begin{lemma} \label{lem:two-reads-2l-dist}
    Let $d$ be odd, let $e = \frac{d-1}{2}$, and let $t = e + \ell$. Also, assume that $\ell < \frac{d}{2}$.
    Let $Y$ be a set containing $N_{n,q}(t, d) + 1$ reads. Then, there exist two reads in $Y$ with distance at least $2\ell - 1$.
\end{lemma}
\begin{proof}
    We will prove that $N_{n,q}(t, d) \geq \text{Vol}_q(\ell - 1, n)$ where $\text{Vol}_q(\ell - 1, n)$ is the size of the Hamming ball of radius $\ell-1$. Note that this proves the claim. 
    Indeed, assume that the distance of every pair of reads is at most $2\ell - 2$. 
    It means that all the $N_{n,q}(t, d) + 1$ reads fit inside a ball of radius $\ell - 1$ which leads to a contradiction.

    Let $x, y$ be two codewords such that $d(x,y) = d$, the minimum distance of the code. Let $u\in \F_q^n$ be such that $d(x, u) = d(y,u) = (d + 1)/2$. 
    Consider, $B_{\ell-1}(u)$, the Hamming ball around $u$ of radius $\ell - 1$. 
    We will prove that $B_{\ell-1}(u) \subseteq B_{t}(x)$ and $B_{\ell-1}(u) \subseteq B_{t}(y)$.
    This would imply that $\text{Vol}_q(\ell - 1, n) = |B_{\ell-1}(u)| \leq |B_t(x) \cap B_t(y)| = N_{n,q}(t,d)$.

    Note that by the triangle inequality, for any $v\in B_{\ell-1}(u)$, it holds that $d(v, x)\leq d(v, u) + d(u,x) = \ell-1 + \frac{d+1}{2} = t$ and an identical argument also yields that $d(v, x) \leq t$ for any $v\in B_{\ell-1}(u)$. 
    Thus, we have proved that $B_{\ell-1}(u)\subseteq B_t(x) \cap B_t(y)$, as desired. 
\end{proof}
\vspace{-0.2cm}
Clearly, we can perform a search that would take $\cO(n\cdot N^2)$ time to find two reads that are at distance $2\ell - 1$ apart. 
We note again that $N$ is relatively big compared to $n$, so we would like to reduce our dependency on $N$. 
A simple observation shows that to find two reads that are at distance $\ell$, we can perform only $\cO(n\cdot N)$ iterations. Indeed, set $\bfu$ to be the first read. Then, evaluate the distance of all the reads from $\bfu$ and take the read that is the farthest apart. 
The correctness of this claim is given in the following simple claim
\begin{claim} \label{clm:two-reads-linear-time}
    Let $Y$ be a set of reads such that there are two reads with distance at least $2\ell - 1$. 
    Let $\bfu \in Y$. There exists a $\bfv \in Y$ such that $d(\bfu, \bfv) \geq \ell$.
\end{claim}
\begin{proof} 
    Let $\bfw$ and $\bfz$ be two reads with $d(\bfw, \bfz) \geq 2\ell - 1$. 
    Then, by the triangle inequality, $d(\bfu, \bfw) + d(\bfu, \bfz) \geq d(\bfw,\bfz) \geq 2\ell - 1$. Thus, either $d(\bfu, \bfw) \geq \ell$ or $d(\bfu, \bfz) \geq \ell$.
\end{proof}
Therefore, our final algorithm is given in Algorithm~\ref{alg:rec-two-reads}.
\vspace{-0.2cm}
\begin{algorithm}
    \SetAlgoLined
    \DontPrintSemicolon
    \SetKwInOut{Input}{input}
    \SetKwInOut{Output}{output}

    \Input{A set of $N$ distinct reads $Y$, and an integer $r$}
    \Output{A codeword $c$}
    Set $\bfu = \bfy_1$ and set $\bfv = \max_{\bfy\in Y} d(\bfu, \bfy)$\;
    Execute Algorithm~\ref{alg:rec-dec-alg} with $Y$, $\{\bfu, \bfv\}$, and $r$\; 
    \caption{Reconstruction using two reads}
    \label{alg:rec-two-reads}
\end{algorithm}
\vspace{-0.35cm}

\begin{proposition}
    Let $\varepsilon > 0$ and Let $n, k, t$ be integers such that 
    \[
        \frac{t}{n} \leq 1 - \sqrt{\frac{k}{n} \cdot \left( 1 - \frac{\ell}{2n} + \varepsilon \right)} \;.
    \]
    Let $\cC$ be an $[n,k]_q$ RS code.
    Let $\bfc \in \cC$ be a codeword and let $Y$ be a set containing $N:=N_{q,n}(t,n-k+1)+1$ vectors such that $\forall \bfy \in Y$, we have $d(\bfc,\bfy)\leq t$.
    Then, applying Algorithm~\ref{alg:rec-two-reads} on $Y$ and $\ceil{1/\varepsilon}$ we get back $\bfc$. Furthermore, the time complexity of the algorithm is $\cO(n\cdot N + n^3\varepsilon^{-6})$.
\end{proposition}
\begin{proof}
    By \Cref{lem:two-reads-2l-dist}, we know that $Y$ consists of two reads such that their distance is at least $2\ell - 1$ and by \Cref{clm:two-reads-linear-time}, we know that the first step in Algorithm~\ref{alg:rec-two-reads} finds two reads that are $\ell$ distance apart in $\cO(n\cdot N)$ time. 
    To ensure that the KV algorithm indeed works and the transmitted codeword is inside the list, we check that $S_M(\bfc) > \sqrt{2(k-1) C(M)}$,  where $M$ is the multiplicity matrix constructed by two reads that are $\ell$-distance apart.
    Indeed, in this case, by plugging in \Cref{prop:dec-two-reads} $\rho = t/n$, $\delta = \ell/t$, we get that 
    \[
    \frac{t}{n} \leq 1 - \sqrt{\frac{k}{n} \cdot \left( 1 - \frac{\frac{t}{n} \cdot \frac{\ell}{t}}{2} + \varepsilon \right)} =  1 - \sqrt{\frac{k}{n} \cdot \left( 1 - \frac{\ell}{2n} + \varepsilon \right)} \;,
    \]
    and that running the KV algorithm on $M$ produces a list of size $\cO(1/(\sqrt{R}\varepsilon))$ with $\bfc$ in the list. 
    Note that as $C(M) = \cO(n\cdot \varepsilon^{-2})$, then the KV algorithm runs in time $\cO(n^3\cdot \varepsilon^{-6})$.
    The final step of Algorithm~\ref{alg:rec-dec-alg} iterates over all codewords in the list and for each one, checks if it is at distance $\leq t$ from all the reads. Thus, since the list is of constant size, this takes $\cO(n\cdot N)$ time. 
    Overall, the complexity of Algorithm~\ref{alg:rec-two-reads} is $\cO(n\cdot N + n^3\cdot \varepsilon^{-6})$ as desired.
\end{proof}

\newpage


\newpage 
\appendix
\combLemma*
\begin{proof}
For $g,h\in \mathbb{Z}_{+}$, we know that 
\begin{align*}
    \binom{g+h+1}{2} = \binom{g+1}{2} + \binom{h+1}{2} + gh.
\end{align*}
Therefore, 
\begin{align*}
    &\sum_{i=1}^{b-1}\left(\binom{1 + \sum_{j=0}^{i-1} a_j}{2} + \binom{1+a_i}{2} \right) \\
    &=  \sum_{i=1}^{b-1}\left(\binom{1 + \sum_{j=0}^{i} a_j}{2} - a_i  \left(\sum_{j=0}^{i-1} a_j\right)\right).
\end{align*}
Note that 
\begin{align*}
    \sum_{i=1}^{b-1}a_i\sum_{j=0}^{i-1}a_j = \sum_{i=0}^{b-2}a_i \left(c-\sum_{j=0}^{i}a_j\right).
\end{align*}
Hence, we have that $\sum_{i=0}^{b-1}\binom{a_i+1}{2}$ is 
\begin{align*}    
    &\binom{c+1}{2} - \sum_{i=1}^{b-1}a_i\sum_{j=0}^{i-1}a_j\\
    &= \binom{c+1}{2} - \frac{1}{2}\left(\sum_{i=1}^{b-1}a_i\sum_{j=0}^{i-1}a_j + \sum_{i=0}^{b-2}a_i \left(c-\sum_{j=0}^{i}a_j\right)\right)\\
    &= \binom{c+1}{2} - \frac{1}{2}\sum_{i=0}^{b-1}a_i(c-a_i). 
\end{align*}
\end{proof}
\end{document}